\documentclass[runningheads]{llncs}

\usepackage[margin=2.5cm]{geometry}
\usepackage[fontsize=11pt]{scrextend}
\usepackage{amsfonts,amsmath,amssymb}
\usepackage{graphicx}
\usepackage[ruled,vlined]{algorithm2e}
\usepackage{quiver}
\usepackage{complexity}
\usepackage[strings]{underscore}

\SetKwInput{KwShared}{Shared}
\SetKwInput{KwInitial}{Initial}

\newcommand\restrict[1]{\raisebox{-.5ex}{$|$}_{#1}}
\newcommand{\ceil}[1]{\lceil {#1} \rceil}
\DeclareMathOperator*{\argmax}{arg\,max}
\DeclareMathOperator{\Ima}{Im}

\def\REVIEW{1}  

\if \REVIEW 1
    \newcommand{\guillermo}[1]{{\color{blue}[Guillermo: #1]}}
    \newcommand{\petr}[1]{{\color{red}[Petr: #1]}}
\else
    \newcommand{\guillermo}[1]{}
    \newcommand{\petr}[1]{}
\fi

\def\A{\ensuremath{\mathcal{A}}}

\def\I{\ensuremath{\mathcal{I}}}
\def\O{\ensuremath{\mathcal{O}}}

\begin{document}
%
\title{On the Bit Complexity of Iterated Memory}
%
\author{Guillermo Toyos-Marfurt \and Petr Kuznetsov}
\authorrunning{G. Toyos-Marfurt and P. Kuznetsov}
%
\institute{LTCI, Institute Polytechnique de Paris, Palaiseau, France
\\ \email{\{guillermo.toyos, petr.kuznetsov\}@telecom-paris.fr}}
\maketitle              
\begin{abstract}
Computability, in the presence of asynchrony and failures, is one of the central questions in distributed computing. 
The celebrated asynchronous computability theorem (ACT) characterizes the computing power of the read-write shared-memory model through the geometric properties of its \emph{protocol complex}: a combinatorial structure describing the states the model can reach via its finite executions.     
This characterization assumes that the memory is of unbounded capacity, in particular, it is able to store the exponentially growing states of the \emph{full-information} protocol.
%

In this paper, we tackle an orthogonal question: What is the \emph{minimal} memory capacity that allows us to simulate a given number of rounds of the full information protocol? 
%
In the \emph{iterated immediate snapshot} model (IIS),
%
we determine necessary and sufficient conditions on the number of bits an IIS element should be able to store so that the resulting protocol is equivalent, up to isomorphism, to the full-information protocol. 
Our characterization implies that $n\geq 3$ processes can simulate $r$ rounds of the full-information IIS protocol as long as the bit complexity per process is $\Theta(r n \log n)$.
Two processes, however, can simulate \emph{any} number of rounds of the full-information protocol using only $2$ bits per process, which implies, in particular, that just $2$ bits per process are sufficient to solve $\varepsilon$-agreement for arbitrarily small $\varepsilon$. 
%



\end{abstract}

\section{Introduction}

One of the central questions in the theory of distributed computing is how to characterize the class of problems that can be solved in a given system model. 
%
There exists a plethora of models, varying synchrony assumptions, fault types, and communication media.
Evaluating their computing power and establishing meaningful comparisons between these models poses a formidable challenge.

%

The combinatorial approach~\cite{HS99,distCompTopo} has proved to be instrumental in analyzing the computing power of a wide range of shared-memory models~\cite{HS99,HR10,GKM14-podc,KRH18}. 
A set of executions produced by a model can be represented as a \emph{simplicial complex}, a combinatorial structure, whose  
geometric properties allow us to reason about the model's computing power. 

For example, the \emph{asynchronous computability theorem} (ACT)~\cite{HS99} characterizes wait-free task solvability of the \emph{read-write shared-memory} model. 
Here, a task is defined as a tuple $(\I,\O, \Delta)$, where $\I$ and $\O$ are, resp., the \emph{input complex} and the \emph{output complex} , describing the task's input and output configurations, and $\Delta: \I \rightarrow 2^{\O}$ is a map relating each input assignment to all output assignments allowed by the task.
A task is solvable, if and only if there exists a continuous map from $\I$ to $\O$ that \emph{respects} $\Delta$.         

Notice that the characterization is defined entirely by the task, i.e., regardless of the ``operational'' behavior of the model.
The core observation behind ACT and its followups~\cite{HR10,GKM14-podc,KRH18} is that finite executions of the read-write model can be described via iterations of the \emph{immediate snapshot}~\cite{BG93a,BG97,GR10-opodis}.
It turns out that the \emph{iterated immediate snapshot} model (IIS)  has a very nice regular structure, precisely captured by recursive applications of the \emph{standard chromatic subdivision}~\cite{Koz12} (Figure~\ref{fig:chromatic}). 
This representation assumes, however, the \emph{full-information protocol}: in each iteration, a process writes its current state to the memory and updates its state with a snapshot of the states written by the other processes.
As a result, the state grows exponentially with the number of iterations, which requires memory locations of unbounded capacity.

In this paper, we consider a more realistic scenario of \emph{bounded} iterated immediate snapshot memory (B-IIS). 
More precisely, we determine necessary and sufficient conditions for the capacity of the shared iterated memory so that its protocol complex is isomorphic to $r$ iterations of the standard chromatic subdivision and thus maintains the properties of $r$ IIS rounds.  
Intuitively, this boils down to defining an optimal \emph{encoding} of a process state at the end of round $r$ using the available number of bits, so that, given a snapshot of such encoded values, one can consistently map these states to the vertices of the $r$ iterations of the standard chromatic subdivision.  
%
One can then derive the conditions on the \emph{bit complexity}, i.e., the amount of memory required to reproduce the \emph{full-information} protocol in the B-IIS model.

Our contributions are the following: 
Given an input complex $\I$ of a task, we determine necessary and sufficient conditions on the number of bits that enable an encoding function for simulating $r$ rounds of (full-information) IIS.  
%
%
%
We show that for $3$ or more processes, a tight asymptotic bound on the per-process bit complexity (i.e., the amount of memory a process must have at its disposal in each iteration) is \emph{linear} by the number of rounds and \emph{linearithmic} by the number of processes: $\Theta(r n \log n)$. 
For two processes, the message complexity is $\Theta(1)$. 
To obtain these results, we establish an insightful combinatorial characterization of the iterated standard chromatic subdivision via $\mathrm{f}$-vectors, a common tool in polyhedral combinatorics~\cite{fvector}.

\subsubsection{Roadmap.}
The paper is organized as follows. 
We overview the related work in Section~\ref{sec:related}. 
In Section~\ref{sect:preliminaries}, we recall the combinatorial-topology notions used in this work.  
In Section~\ref{sect:System model}, we describe the B-IIS model. 
Then in Section~\ref{sect:Equivalence_section}, we give the conditions that the B-IIS model has to satisfy in order to replicate the iterated full information protocol.
Next, in Section~\ref{sect:fvector}, we characterize these conditions by studying the $\mathrm{f}$-vector of the iterated chromatic subdivision. 
Section~\ref{sect:fvector_asympt} contains the proofs of the asymptotic bounds derived from the equations in Section~\ref{sect:fvector}. 
Finally, in Section~\ref{sect:final results}, we give the asymptotic bounds of bit complexity of the B-IIS model and a discussion of the obtained results.
The special case of a $2$-process system and the approximate agreement task is relegated to Appendix~\ref{sect:approx_ag}. 

\section{Related Work}
\label{sec:related}

The fundamental asynchronous computablity theorem (ACT) was established by Herlihy and Sha-vit~\cite{HS99},  using elements of simplicial topology.
This characterization enabled a solution to the long-standing conjecture on the impossibility of wait-free set agreement, concurrently derived by Saks and Zacharoglou~\cite{SZ93}, and Borowsky and Gafni~\cite{BG93a}.  
The characterization~\cite{HS99} applied to wait-free read-write memory was later extended to more general classes of iterated models~\cite{HR10,GKM14-podc,KRH18}. 
Herlihy, Kozlov and Rajsbaum~\cite{distCompTopo} later gave a thorough introduction to the use of combinatorial topology in distributed computing. 
%


Borowsky and Gafni introduce the simplex approximation task and presented the first wait-free read-write protocol for simplex approximation on the iterated standard chromatic subdivision.      
Nishimura introduced a variant on the IIS model which optimizes the amount of layers needed for solving a task by simplex approximation~\cite{nishi}. 
Hoest and Shavit~\cite{hoest} presented an alternative IIS model to study its execution time complexity.

Communication complexity is a well-studied sub-field of complexity theory, which considers the amount of information required for multiple distributed agents to acquire new knowledge~\cite{comm_cmplx}. 
Yet, the problem of bit complexity of full-information protocols has only been addressed recently. 
Delporte-Gallet et. al.~\cite{2process_complexity} show that for the particular case of two processes, it is possible to solve any wait-free solvable task using mostly $1$-bit messages in wait-free dynamic networks. 
Notably, the algorithm for achieving a $2$-process  $\epsilon$-agreement has the same characteristics as the one presented in this work. It can be said that our contribution lies in delivering a communication-efficient  $\epsilon$-agreement within the IIS protocols. 
Another version of $1$-bit approximate agreement on dynamic networks is presented in~\cite{epistem}.

Recently, Delporte-Gallet et. al.~\cite{delporte2023computational} gave a comprehensive study of the computational power of memory bounded registers. 
Using combinatorial tools, they determine bit complexity required for the existence of wait-free and $t$-resilient read-write solutions for a given task. 
However, their approach solely focuses on task solvability, relying on simulations and other techniques that result in protocols requiring an exponential number of rounds for simulating a single round of IIS. 
In this paper, we aim to simulate IIS using \emph{the same} number of rounds,
thus, ensuring the same time complexity.

\section{Preliminaries} \label{sect:preliminaries}

\subsubsection{Simplicial complex.} 
We represent the IIS and B-IIS model as topological spaces, defined as \emph{simplicial complexes}~\cite{distCompTopo}: 
a set of \emph{vertices} and an inclusion-closed set of vertex subsets, called \emph{simplices}. 
We call the simplicial complex of $n+1$ vertices with its power set the $n$-dimensional simplex $\Delta^n$. The \emph{dimension} of a simplex $\Delta$, denoted as $dim(\Delta)$, is its number of vertices $V(\Delta)-1$.
The dimension of a simplicial complex is equal to the dimension of the largest simplex it contains.
We call \textit{faces} the simplices contained in a simplicial complex, where a $0$-face is a vertex. 
We call \textit{facets} the simplices which are not contained in any other simplex. We denote the set of vertices of a simplicial complex $\mathcal{A}$ as $V(\mathcal{A})$.
We call a simplicial complex \emph{chromatic} when it comes with a \emph{coloring function} that assigns every vertex to a unique process identifier. 
When considering vertices, we use the notation $\restrict{p}$ when we consider only the vertices of color $p$ and we denote $v_p$ a vertex with color $p$.
%
%
\subsubsection{Standard constructions.}
Given a simplicial complex $\mathcal{A}$, the \emph{star} of $\mathcal{S}\subseteq\mathcal{A}$, $\mathrm{St}(\mathcal{A},\mathcal{S})$, is a subcomplex made of all simplices in $\mathcal{A}$ containing a simplex of $\mathcal{S}$ as face. 
The \emph{link}, $\mathrm{Lk}(\mathcal{A},\mathcal{S})$ is the subcomplex consisting of all simplices in $\mathrm{St}(\mathcal{A},\mathcal{S})$ which do not have a common vertex with $\mathcal{S}$.

\subsubsection{Standard chromatic subdivision.}
The \emph{standard chromatic subdivision} of a simplicial complex $\mathcal{A}$, denoted as $\mathrm{Ch}\ \mathcal{A}$, is a complex whose vertices are tuples $(c,\sigma)$ where $c$ is a color and $\sigma$ is a face of $\mathcal{A}$ containing a vertex of color $c$. 
A set of vertices of $\mathrm{Ch}\ \mathcal{A}$ defines a simplex if for each pair $(c,\sigma)$ and $(c',\sigma')$, $c\neq c'$ and either $\sigma \subseteq \sigma'$ or $\sigma' \subseteq \sigma$. 
For $\Delta^n$, the vertices $(c,\Delta^n)$ define a \emph{central simplex} in $\mathrm{Ch}\ \mathcal{A}$. 
The standard chromatic subdivision can be seen as the ``colored'' analog of the \emph{standard barycentric subdivision}~\cite{distCompTopo}.
Figure~\ref{fig:chromatic} illustrates the application of $\mathrm{Ch}$ to the $2$-dimensional simplex.
In general, a subdivision operator $\tau : \A \rightarrow B$ is a map that ``divides'' the simplices in $\A$ into smaller simplices. For a rigorous definition see~\cite{distCompTopo}. 

\begin{figure}
    \centering
    \includegraphics[width=0.6\columnwidth]{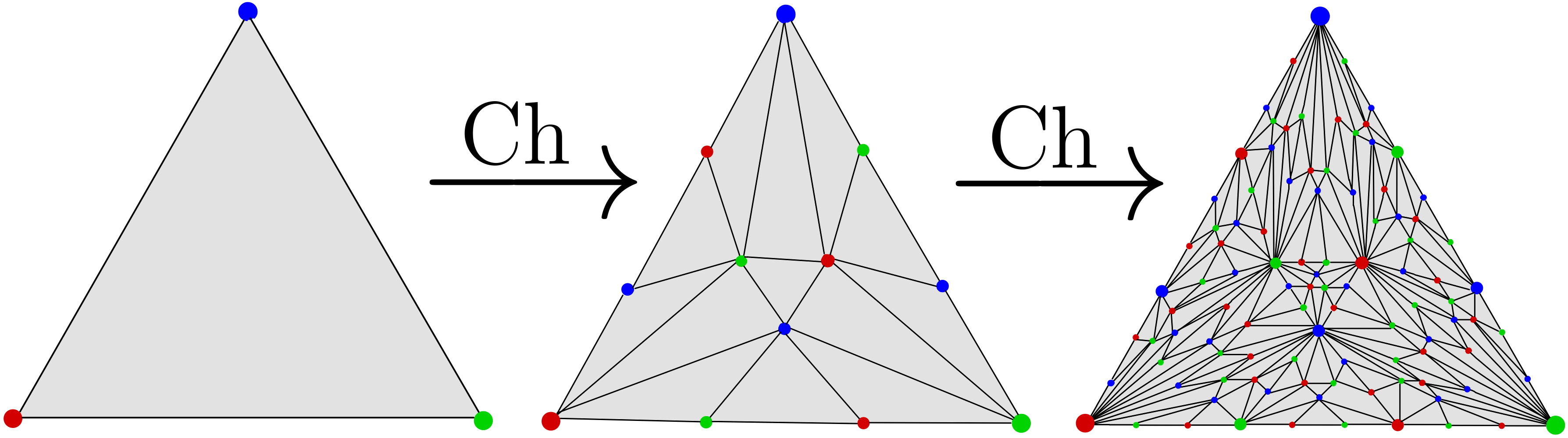}
    \caption{Iterative application of the standard chromatic subdivision to $\Delta^2$. Leftmost diagram illustrate the $\Delta^2$ simplex. The diagram in the middle shows $\mathrm{Ch}\ {\Delta^2}$, which is subdivided again resulting in the rightmost image $\mathrm{Ch}^2\ {\Delta^2}$.}
    \label{fig:chromatic}
\end{figure}

\subsubsection{Tasks.} 
A \emph{distributed task} (or simply a \emph{task}) is defined as a tuple $(\I,\O, \Delta)$, where $\I$ and $\O$ are, resp., the \emph{input complex} and the \emph{output complex}, describing the task's input and output configurations, and $\Delta: \I \rightarrow 2^{\O}$ is a map relating each input assignment to all output assignments allowed by the task.
In this work we consider general colored tasks, where $\I$ and $\O$ are chromatic simplicial complexes, allowing processes to have different sets of inputs and outputs.
To solve a task, a distributed system uses a (communication) \emph{protocol} to share information between processes - \emph{IIS model} is an example of such a protocol. 
The different possible configurations reached after communication can be characterized by a protocol map $\Xi$. From the protocol complex $\Xi(\mathcal{I})$, a decision function $\delta : \Xi(\mathcal{I}) \rightarrow \mathcal{O}$ maps the simplices from the protocol complex to outputs, respecting the specification of the problem given by $\Delta$: $\delta \circ \Xi(\mathcal{I}) \subseteq \Delta(\mathcal{I})$. 
For a formal treatment of task solvability using these concepts, please refer to~\cite{distCompTopo}.

\section{System Model} \label{sect:System model}

\subsection{Computational model} \label{sect:Computational Model}

The IIS protocol is  \emph{wait-free}, i.e., each process is expected to reach an output in a finite number of its own steps, regardless of the behavior of other processes.
Computation is structured in a sequence of \emph{immediate-snapshot layers} $M[1],\ldots,M[r]$, where each layer is represented as an array of shared variables (one per process), initialized to $\bot$. 
%
%
Algorithm~\ref{alg:IIS} shows the pseudo-code of the IIS protocol.

\begin{algorithm}
\SetAlgoLined
\KwShared{$M[r]$ array of $r$ memory-bounded snapshots with $|\Pi|$ entries.}
\KwInitial{$v = input(i)$ $\triangleright$ What the process sees. At first, its input.}
\For{$k:=1$ to $r$}{
    $M[k,i]$ $\gets$ $v$\;
    $v$ $\gets$ $snapshot(M[k])$\;
}
\Return $\delta_i(v)$
\caption{Iterated Immediate Snapshot (IIS) protocol. Code for process $p_i \in \Pi$ and $r > 0$ rounds.}
\label{alg:IIS}
\end{algorithm}

Note that the protocol is \emph{full-information} and \emph{generic}: in each layer, a process writes its \emph{complete view} and then updates its view as a snapshot of views written by (a subset of) other processes.
After a certain number of rounds, the process applies the task-specific decision function $\delta_i(v)$ based on the process's view $v$, to compute the task's output. 
%
%

Algorithm~\ref{alg:B-IIS} shows the bounded IIS communication model, the pseudo-code is written in the same style as the full-information protocol. 
In B-IIS, instead of writing their entire state, processes write a representation of their states using an \emph{encoding function}. 
Then, depending on what the process reads on the next snapshot, the internal state changes according to the $next\_state_i$ function. 
Observe that each process keeps track of their state by using a local variable $s$. 
Finally, the decision function $\delta_i$ now takes the internal state of the process.

\begin{algorithm}
\SetAlgoLined
\KwShared{$M[r]$ array of $r$ memory-bounded snapshots with $|\Pi|$ entries.}
\KwInitial{$v = input(i)$ $\triangleright$ What the process sees. At first, its input.}
\KwInitial{$s = next\_state_i(v,\bot,0)$ $\triangleright$ Process state. At first, determined by the process input.}
\For{$k:=1$ to $r$}{
    $M[k,i]$ $\gets$ $encode_i(s,k)$\;
    $v$ $\gets$ $snapshot(M[k])$\;
    $s$ $\gets$ $next\_state_i(v,s,k)$\;
}
\Return $\delta_i(s)$
\caption{Bounded Iterated Immediate Snapshot (B-IIS) protocol. Code for process $p_i \in \Pi$ and $r > 0$ rounds.}
\label{alg:B-IIS}
\end{algorithm}

Consequently, depending on the $\textit{encode}_i$ and $next\_state_i$ functions we have different communication protocols. 
Note that removing the round parameter $k$ on the $\textit{encode}$ and $\textit{next\_state_i}$ functions does not restrict the expressiveness of the algorithm, as the round number can be inferred from the internal state $s$.
A trivial observation is that if both the encoding and the state function are the identity (that is, $\textit{encode}_i: state,round \mapsto state$ and $\textit{next\_state}_i: view,state,round \mapsto view$) the protocol is equal to the original (full-information) IIS and the bit complexity is $|\Pi|^{r}$ in the $r$-th round. 
(Here $\Pi$ is the set of processes.) 
Thus, the bounds on bit complexity can be imposed by restricting the output space of the $\textit{encode}_i$ functions: $|\Ima \textit{encode}_i | < B$, where $\Ima \textit{encode}_i$ is the image of $\textit{encode}_i$. We also call $\Ima \textit{encode}_i$ the \emph{encoding set}, denoted as $E$.
Therefore, we need to find an ``intelligent'' $\textit{encode}$ function that uses as few bits as possible to encode a state $s$, while still being able to be decoded by the $next\_state_i$ functions of the other processes. 
In Section~\ref{sect:Equivalence_section} we will show necessary and sufficient conditions which allow us to determine exactly how many bits are needed.

\subsection{Topological model}

We consider a system with $n+1$ asynchronous processes $\Pi=\{p_1,\dots,p_{n+1}\}$, where any processes can fail at any time. Processes communicate using the IIS and B-IIS protocol as described in Section~\ref{sect:Computational Model}. We denote $\Xi$ and $\Xi_b$ the protocol maps of IIS and B-IIS respectively. As both protocols are iterative we can denote $\Xi^r$ and $\Xi_b^r$ for the protocol maps of the $r-th$ iteration of algorithm~\ref{alg:IIS} and~\ref{alg:B-IIS} respectively. Note that $\Xi_b$ depends on the $\textit{encode}$ and $next\_state$ functions. When not stated explicitly, we assume these functions are well defined and have the same behavior as the full-information protocol. That is, the protocol complex $\Xi_b(\mathcal{I})$ is isomorphic to $\Xi(\mathcal{I})$.

\section{Equivalence between Bounded and Full-Information IIS} \label{sect:Equivalence_section}

The goal of this work is to explore how, by leveraging the B-IIS model, we can attain a protocol complex equivalent to the (full-information) IIS. In other words, we aim to identify the conditions that $\textit{encode}$ and $next\_state$ functions must satisfy for B-IIS to possess the same computational power as IIS, and to determine the memory requirements for achieving this equivalence.

The central idea for achieving this is \emph{process distinguishability}. 
Our objective is for a process to be capable of deducing the state of another process by reading the values stored in shared memory. 
The underlying intuition is that when a process knows the potential configurations within the protocol complex, it can infer the possible states of the other processes. 
Consequently, if the other processes can effectively communicate their respective states, there is no need to write their entire state in shared memory. 
Instead, a signal suffices for identification within the protocol complex. The following definitions formalize such notions:

\begin{definition}[Distinguishability of a vertex] \label{def:disting}
    Let $\mathcal{A}$ a chromatic complex, $\textit{encode} : V(\mathcal{A})\rightarrow E$, $p, q \in \Pi$ and $s_p, t_q$ two adjacent vertices in $\mathcal{A}$. We say that $s_p$ \textbf{distinguishes} $t_q$ under $\textit{encode}$ if there is no other vertex $w_q \in \mathrm{Lk}(\mathcal{A},s_p)\arrowvert_q$ sharing the same encoding: $\textit{encode}(t_q)=\textit{encode}(w_q)$.
\end{definition}

\begin{definition}[Distinguishability of a chromatic simplicial complex] \label{def:disting_simplex}
   Let $\mathcal{A}$ a chromatic complex and $\textit{encode} : V(\mathcal{A})\rightarrow E$. 
   We say that $\mathcal{A}$ is \textbf{distinguishable} under $\textit{encode}$ if for every pair of adjacent vertices $s,t \in V(\mathcal{A})$, $s$ distinguishes $t$ and vice versa.
\end{definition}

In other words, Definition~\ref{def:disting} stipulates that a vertex in the protocol complex is distinguishable by another vertex if the former can tell apart the latter just by reading the written value in the snapshot memory.
Then, if this relation is true for all vertices on the simplicial complex, we say that the whole complex is distinguishable.

Theorem~\ref{th:equiv_xib} establishes that when the input complex is distinguishable, the protocol complex of B-IIS exhibits an identical structure to the full-information protocol, namely the standard chromatic subdivision.

\begin{theorem}[Characterization of $\Xi_b$] \label{th:equiv_xib}
    Let $\mathcal{I}$ a chromatic input complex, $\textit{encode} : V(\mathcal{A})\rightarrow E$ the encoding function of the B-IIS protocol, and $\Xi_b$ its associated protocol map. The following equivalence holds:

    \[
    \Xi_b(\mathcal{I}) \cong \mathrm{Ch} \ \mathcal{I} \iff \mathcal{I} \text{ is distinguishable under } \textit{encode}
    \]
\end{theorem}
\begin{proof}
    For the sufficient condition, we will show that $\Xi_b(\mathcal{I}) \cong \Xi(\mathcal{I})$, as we have that the protocol map of the unbounded IIS protocol is equivalent to the standard chromatic subdivision: $\Xi \cong \mathrm{Ch}$. Take any facet $\sigma$ of $\mathcal{I}$, let $v_p, w_q \in V(\sigma)$. Because $v_p$ distinguishes $w_q$, there exists a locally invertible function $decode$ for each state of $p$ such that $decode_p(\textit{encode}_q(w_q),v_p)=w_q$. Thus if $next\_state_p$ has the form $view,state \mapsto decode_p(view,state)$, then $next\_state_p$ is injective in $view$ and returns the state of each vertex in $\mathrm{Lk}(\mathcal{I},v_p)$. Next, if the values written by $\textit{encode}$ makes all pair of vertices distinguishable, then, from all states in $\mathcal{I}$, the behavior of $\Xi_b$ is the same: from a snapshot, it decodes the views of the other processes, obtaining their respective local states. Thus, $\Xi_b(\mathcal{I})$ has the same behavior as the $\textit{full information}$ protocol $\Xi$: from reading the snapshot we get the views of the other processes. Thus $\Xi_b(\mathcal{I}) \cong \Xi(\mathcal{I}) \implies \Xi_b(\mathcal{I}) \cong \mathrm{Ch}\ \mathcal{I}$.

    Now for the necessary condition we will prove it by contradiction. Take an input complex composed of two processes $p,q$ where the first process has one state which is adjacent to two different states of the second process. That is $\mathcal{I}=\{v_p,w_q,t_q,\{v_p,w_q\},\{v_p,t_q\},\{\}\}$. And suppose that $\Xi_b(\mathcal{I}) \cong \mathrm{Ch}\ \mathcal{I}$ but $v_p$ does not distinguish $w_q$ nor $t_q$ as $\textit{encode}_q(w_q)=\textit{encode}_q(t_q)=e$, making $\mathcal{I}$ not distinguishable under $\textit{encode}$.

    There are only two possibilities for process $p$: It reads the value $e$ encoded by both states of $q$ or it reads $\bot$. For $q$, there are two possibilities from each state: $q$ reads the value written by $p$ or not. Therefore, the states where $p$ reads $e$ with all four states of $q$ are compatible configurations. Then $\Xi_b(\mathcal{I})$ will have a vertex of degree 4, which implies that $\Xi_b(\mathcal{I}) \ncong\mathrm{Ch}\ \mathcal{I}$ - a contradiction. $\qed$
\end{proof}

Theorem~\ref{th:equiv_xib} gives the requirements that the B-IIS protocol must satisfy, in terms of distinguishability, for the protocol complex to be equivalent to the full information one. Note that Theorem~\ref{th:equiv_xib} takes an arbitrary input complex. Thus, we can iteratively apply the theorem to obtain the protocol complex after an arbitrary number of rounds. 

\begin{corollary}\label{coro:equiv_B-IIS_iterated}
    Let $\mathcal{I}$ be a chromatic input complex, $r\geq 1$, and $\textit{encode} : V(\mathcal{A})\rightarrow E$ the encoding function of the B-IIS protocol. The following equivalence holds:
    \[
    \Xi_b^r(\mathcal{I}) \cong \mathrm{Ch}^r \ \mathcal{I} \iff \forall r' \in [0, r-1], \mathrm{Ch}^{r'}\ \mathcal{I} \text{ is distinguishable under } \textit{encode}
    \]
\end{corollary}

Figure~\ref{fig:iso-diag} shows a commutative diagram illustrating the process. Starting from the initial input complex $\mathcal{I}$, we can equivalently apply $r$ times either $\Xi_b$ or $\mathrm{Ch}$ to yield the $r$-th round protocol complex. Subsequently, the decision function is applied yielding the output complex.

\begin{figure}
\centering
\[\begin{tikzcd}
	& {\Xi_{b}(\mathcal{I})} && {\Xi_{b}^2(\mathcal{I})} && \dots && {\Xi_{b}^r(\mathcal{I})} \\
	{\mathcal{I}} &&&&&&&& {O} \\
	& {\mathrm{Ch}\ \mathcal{I}} && {\mathrm{Ch}^2\ \mathcal{I}} && \dots && {\mathrm{Ch}^r\ \mathcal{I}}
	\arrow["{\Xi_{b}}", from=2-1, to=1-2]
	\arrow["\mathrm{Ch}\ "', from=2-1, to=3-2]
	\arrow["{\mathrm{Ch}}"', from=3-2, to=3-4]
	\arrow["{\mathrm{Ch}}"', from=3-4, to=3-6]
	\arrow["{\mathrm{Ch}}"', from=3-6, to=3-8]
	\arrow["{\Xi_{b}}", from=1-2, to=1-4]
	\arrow["{\Xi_{b}}", from=1-4, to=1-6]
	\arrow["{\Xi_{b}}", from=1-6, to=1-8]
	\arrow["\delta", from=1-8, to=2-9]
	\arrow["\delta"', from=3-8, to=2-9]
	\arrow[dashed, tail reversed, from=3-2, to=1-2]
	\arrow[dashed, tail reversed, from=1-4, to=3-4]
	\arrow[dashed, tail reversed, from=3-8, to=1-8]
	\arrow[dashed, tail reversed, from=3-6, to=1-6]
\end{tikzcd}\]
\caption{Commutative diagram illustrating the equivalence between $\Xi_b$ and $\mathrm{Ch}$ over an input complex $\mathcal{I}$ after $r$ iterations. Dashed arrows indicate the existence of an isomorphism.}
\label{fig:iso-diag}
\end{figure}
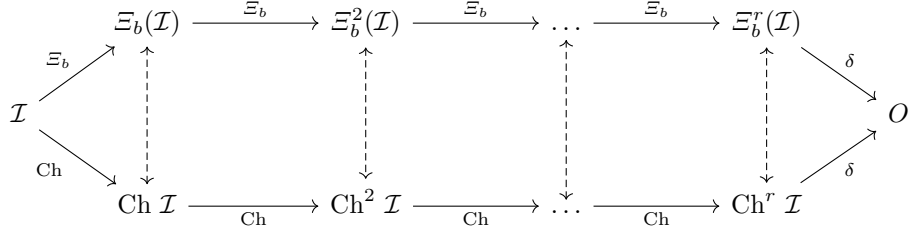

We will further develop the requirements for the encoding function to ensure distinguishability of a simplicial complex. To accomplish this, we define the indistinguishability graph. The idea is that for each process we get a graph that describes the vertices that need to have a different encoding for the simplicial complex to be distinguishable under an arbitrary encoding function. Figure~\ref{fig:disting_graph} presents an illustrative instance of a simplicial complex and its associated indistinguishability graph for one of the processes.

\begin{definition}[Indistinguishability graph of a chromatic complex]
    Let $\mathcal{A}$ chromatic simplicial complex, $p \in \Pi$ and $ v,w \in V(\mathcal{A})\restrict{p} : v\neq w$. We define the \textbf{indistinguishability graph} $G_p : \mathcal{A} \times \Pi \rightarrow (V(\mathcal{A}),V(\mathcal{A})\times V(\mathcal{A}))$  of $\mathcal{A}$ with respect to process $p$ as follows:
    \[
    (v,w) \in G_p(\mathcal{A}) \iff \exists t\in V(\mathcal{A}) : v,w \in V(\mathrm{Lk}(\mathcal{A},t))\restrict{p}
    \]
\end{definition}

\begin{figure}
    \centering
    \includegraphics[width=0.6\columnwidth]{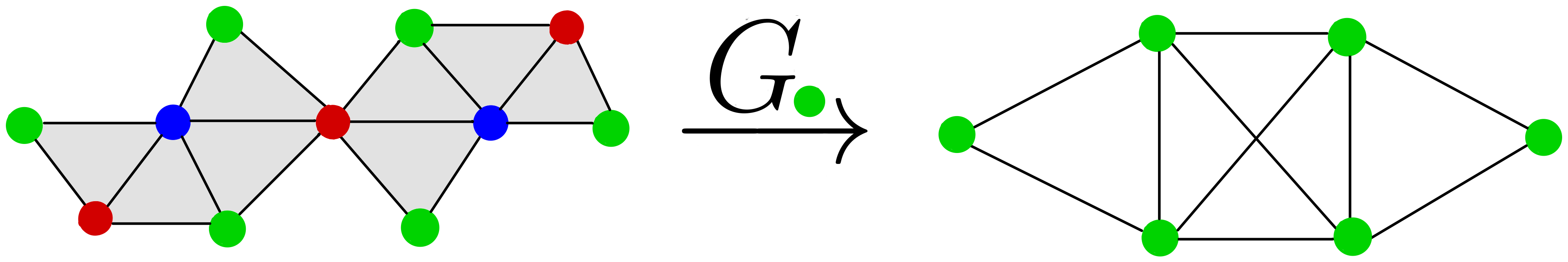}
    \caption{Example of a simplicial complex along with its corresponding indistinguishability graph of the green-labeled process. The relative positions of the green vertices are preserved in the drawing.} 
    \label{fig:disting_graph}
\end{figure}

The following theorem characterizes the problem of determining whether an encoding renders an input complex distinguishable. Additionally, it addresses the challenge of finding an encoding function for an arbitrary simplicial complex. This is achieved by establishing the equivalence between these problems and the task of finding a vertex coloring for the indistinguishability graphs associated with a simplicial complex.

\begin{theorem} \label{theorem:vertex colouring}
Let $\mathcal{A}$ be a chromatic input complex with process labelling in $\Pi$ and  $\textit{encode} : V(\mathcal{A})\rightarrow E$ an encoding function,
    \[
    \mathcal{A} \text{ is distinguishable under } \textit{encode} \iff \forall p \in \Pi, \textit{encode} \text{ is a proper vertex coloring of } G_p(\mathcal{A}) 
    \]
\end{theorem}
\begin{proof}
We will prove the necessary and sufficient conditions by contradiction. For the former, suppose that there is an edge $(v_p,w_p) \in G_p(\mathcal{A})$ such that $\textit{encode}(v_p)=\textit{encode}(w_p)$. Because such edge is in $G_p(\mathcal{A})$, by definition of indistinguishability graph, it implies that there exists $t_q \in V(\mathcal{A})$ such that $t_q$ is adjacent to $v_p$ and $w_p$ in the protocol complex. However, as both vertices have the same value under $\textit{encode}$, $t_q$ cannot distinguish $v_p$ and $w_p$. A contradiction, as we assumed that $\mathcal{A}$ is distinguishable under $\textit{encode}$. 
Now for the sufficient condition, suppose that $\mathcal{A}$ is not distinguishable but $\textit{encode}$ is a proper vertex coloring of $G_p(\mathcal{A})$. That implies that there is a vertex $t_q$ that is adjacent to $w_p$ and $v_p$ such that $\textit{encode}(v_p)=\textit{encode}(w_p)$. As $w_p,v_p \in \mathrm{Lk}(\mathcal{A},t_q)$ we have that $(w_p,v_p)$ in $G_p(\mathcal{A})$. A contradiction, because both vertices have the same color in $\textit{encode}$. $\qed$
\end{proof}

As an outcome of Theorem~\ref{theorem:vertex colouring}, we can now provide insights into the amount of information that the encoding function needs to write by leveraging well-known results from graph theory. Hence, the bit complexity required to satisfy Theorem~\ref{th:equiv_xib} corresponds to the maximum chromatic number among all graphs in $\{G_p(\mathcal{A})\}_{p\in\Pi}$. The encoding function does not have to encode any information about the process state but a distinct value that renders the process state distinguishable within the protocol complex. Then, the $\textit{next\_state}$ function can decode the message using the coloring of the input complex. Thus, the amount of information to write can be expressed in terms of the cardinality of the encoding set: $|\Ima \textit{encode}|$. It follows that to implement the encoding function we will need at least $\log_2(|\Ima \textit{encode}|)$ bits, giving the bit complexity of the B-IIS protocol.

\begin{corollary}\label{coro:ima_lowbound}
    Let $\mathcal{A}$ a chromatic simplicial complex, $\textit{encode} : V(\mathcal{A})\rightarrow E$ encoding function and $\omega(G)$ the size of the largest clique in $G$. The following condition gives a lower bound on the size of the encoding set:
    \[
    \mathcal{A} \text{ is distinguishable under } \textit{encode} \implies |\Ima \textit{encode}| \geq \max_{p \in \Pi}{\omega(G_p(\mathcal{A}))}
    \]
\end{corollary}
\begin{proof}
If $\mathcal{A}$ is distinguishable under $\textit{encode}$, then by Theorem~\ref{theorem:vertex colouring}, $\textit{encode}$ is a proper vertex coloring of the family of graphs $G_\Pi(\mathcal{A})=\{G_p(\mathcal{A})\}_{p\in\Pi}$. Let $\omega$ be the size of the biggest clique of the graphs in $G_\Pi(\mathcal{A})$. In order to color a clique of size $\omega$, we need at least $\omega$ different colors. As a result, the cardinality of the encoding set - $|\Ima \textit{encode}|$ - has to be at least as big as $\omega$. $\qed$
\end{proof}
\begin{corollary}\label{coro:ima_highbound}
    Let $\mathcal{A}$ chromatic simplicial complex and $\Delta(G)$ the biggest vertex degree in $G$. There exists an encoding function $\textit{encode} : V(\mathcal{A})\rightarrow E$ such that $\mathcal{A}$ is distinguishable under $\textit{encode}$ and $|\Ima \textit{encode} | \leq \max_{p \in \Pi}{\Delta(G_p(\mathcal{A}))}+1$.
\end{corollary}
\begin{proof}
    Let $\textit{encode}$ be an arbitrary encoding function, in order for $\mathcal{A}$ to be distinguishable under $\textit{encode}$, we require $\textit{encode}$ to be a proper vertex coloring of all graphs in $\{G_p(\mathcal{A})\}_{p\in\Pi}$. By Brook's theorem~\cite{brooks_1941}, the chromatic number of a graph $G$ is less or equal to $\Delta(G)+1$. That means that there exists a proper vertex coloring function $\textit{encode}$ whose size of its encoding set $|\Ima \textit{encode}|\leq \Delta(G) + 1$. Let $\delta_p$ be a proper coloring of $G_p(\mathcal{A})$, because each $G_p$ takes the vertices of different processes, we can define the $\textit{encode}$ by using each $\delta_p$ for the vertices of process $p$. Then, we have that $\textit{encode}$ is a proper coloring of all the graphs in $\{G_p(\mathcal{A})\}_{p\in\Pi}$ and by Theorem~\ref{theorem:vertex colouring}, $\mathcal{A}$ is distinguishable under \textit{encode}. Moreover, by Brook's theorem $|\Ima \delta_p| \leq \Delta(G_p(\mathcal{A})) + 1$. Thus, the cardinality of the encoding set will be as small as the biggest $\Delta(G_p(\mathcal{A})) + 1$. $\qed$
\end{proof}

It follows directly from Theorem~\ref{theorem:vertex colouring} that finding the smallest encoding set is computationally intractable. As it is equivalent to determining the chromatic number of a graph. Note that the hardness arises from the input complex, which can be arbitrarily large.

\begin{corollary}
Given $\mathcal{A}$ chromatic simplicial complex, the optimization problem of finding a function $\textit{encode}$ such that $\mathcal{A}$ is distinguishable under $\textit{encode}$ and $|\Ima \textit{encode}|$ is minimal is $\NP$-Hard.
\end{corollary}
\begin{proof}
    We will show that given a graph $H$, we can build a chromatic input complex $\mathcal{H}$, such that $G_{p_0}(\mathcal{H})=H$. Thus, the set of indistinguishability graphs is the set of all graphs.

    For each edge $(v_i,v_j) \in H$, we build the simplicial complex $\mathcal{H}_{ij}=\{\{(p_0,v_i),(p_1,v_{ij})\},$
    
    $\{(p_1,v_{ij}),(p_0,v_j)\},\{(p_0,v_i)\},\{(p_1,v_{ij})\},\{(p_0,v_j)\},\emptyset\}$. Then we build $\mathcal{H}$ by doing the union of all these simplices: $\mathcal{H}=\bigcup_{(v_i,v_j) \in H}\mathcal{H}_{ij}$. That is, we build a 2-process simplicial complex by adding a pair of edges for each edge in $H$. Then, we have that $(p_0,v_i),(p_0,v_j) \in V(\mathrm{Lk}(\mathcal{H},(p_1,v_{ij})))\ \forall i,j$. As these are the only vertices in the link, we have that $H=G_{p_0}(\mathcal{H})$.

    As a result, the optimization problem requires determining the chromatic number of arbitrary big graphs, which is $\NP$-Hard. $\qed$
    
\end{proof}
\section{$\mathrm{f}$-vector Analysis for Distinguishability and Iterative Chromatic Subdivision}\label{sect:fvector}

Corollaries~\ref{coro:ima_lowbound} and~\ref{coro:ima_highbound} establish lower and upper bounds on the cardinality of the encoding set $|\Ima{\textit{encode}}|$ w.r.t the clique size and degree of the indistinguishability graphs. We can now forget about the task to solve, concurrency and the B-IIS protocol and define the information needed to be written as a counting problem over a topological space: the iterated standard chromatic subdivision of a chromatic simplicial complex.

First, in Section~\ref{sect:fvector_def} we give the necessary definitions used throughout this section. Then in Section~\ref{sect:fvector_disting}, we establish the relationship between the clique size and vertex degree of indistinguishability graphs and their associated simplicial complex. Then in Section~\ref{sect:fvector_chroma} we investigate the behavior of the $\mathrm{f}$-vector under chromatic subdivision. In Section~\ref{sect:fvector_asympt}, we conduct an asymptotic analysis of the equations obtained to derive asymptotic lower and upper bounds on the size of the encoding set.

\subsection{Definitions} \label{sect:fvector_def}

\subsubsection{Additional standard constructions on simplicial complexes.} 
Given a simplicial complex $\mathcal{A}$, the \emph{open star} of $\mathcal{S}\subseteq\mathcal{A}$, $\mathrm{St}^\circ(\mathcal{A},\mathcal{S})$ is defined as: $\mathrm{St}^\circ(\mathcal{A},\mathcal{S})=\{\sigma \in \mathcal{A} : \mathcal{S} \subseteq \sigma \}$. Note that the open star does not yield a simplicial complex but just a topological space. 
The boundary of $\mathcal{A}$ is defined as  $\partial(\mathcal{A}) = \{\sigma \in \mathcal{A} : \sigma \subset G \text{ for a unique facet } G \in \mathcal{A} \} \cup \{\emptyset\}$. Then we can define the interior of $\mathcal{A}$ as $Int\  \mathcal{A}=\mathcal{A} \setminus \partial(\mathcal{A})$. Observe that the boundary is a simplicial complex but the interior is not.
The simplicial join of two simplicial complexes $\mathcal{A}$ and $\mathcal{B}$, denoted as $\mathcal{A} * \mathcal{B}$ is the simplicial complex with set of vertices $V(\mathcal{A})\cup V(\mathcal{B})$ whose faces are all the unions $\alpha \cup \beta,\alpha \in \mathcal{A}$ and $\beta \in \mathcal{B}$.
Given a non-negative integer $l$, the $l$-skeleton of a simplicial complex $\mathcal{A}$, denoted as $skel^l(\mathcal{A})$, is the set of simplices of $\mathcal{A}$ with dimension at most $l$. In particular, $skel^0(\mathcal{A})=V(\mathcal{A})$.

\subsubsection{$\mathrm{f}$-vector.} 
The \emph{$\mathrm{f}$-vector} of a simplicial complex $\mathrm{f}(\mathcal{A})=(\mathrm{f}_{-1}(\mathcal{A}),\mathrm{f}_0(\mathcal{A}),\dots,\mathrm{f}_n(\mathcal{A}))$ has the number of $k$-dimensional faces in its $k$-th coordinate: $\mathrm{f}_k(\mathcal{A})=|\{\sigma \in \mathcal{A} : dim(\sigma) = k \}|$ for $-1 \leq k \leq n$. 
In particular, if $\mathcal{A}\neq \emptyset$ we have that $\emptyset \in \mathcal{A}$ and $f_{-1}(\mathcal{A})=1$. 
Enumerating the faces of a simplicial complex is a recurrent problem in algebraic combinatorics and discrete geometry. An overview of these kind of problems is presented in~\cite{Klee2016}.

\subsection{$\mathrm{f}$-vector and the Indistinguishability Graphs}\label{sect:fvector_disting}

\begin{lemma}\label{lemma:clique_fvector}
    Let $\mathcal{A}$ be a chromatic simplicial complex, the following is a bound on the clique number of the indistinguishability graphs:
    \[
    \max_{p \in \Pi}{\omega(G_p(\mathcal{A}))} \geq \max_{p \in \Pi}{\max_{v \in V(\mathcal{A})}{\mathrm{f}_1(\mathrm{St}^\circ(\mathcal{A},v)\restrict{p})}}
    \]
\end{lemma}
\begin{proof}
    Let $v_p \in V(\mathcal{A})$ such that $v_p$ has the biggest number of adjacent vertices in $\mathcal{A}$ of a single process $q \in \Pi$. This can be enumerated by taking the $\mathrm{f}$-vector of the open star of $v_p$, and taking only the 1-faces, which by definition is the number of edges that $v_p$ has. Note that in $\mathcal{A}$, all vertices adjacent to $v_p$ of  process $q$ are contained in $\mathrm{Lk}(\mathcal{A},v_p)\restrict{q}$, then by definition of indistinguishability graph, all vertices of $\mathrm{Lk}(\mathcal{A},v_p)\restrict{q}$ will form a clique in $G_p(\mathcal{A})$ with the cardinality of the set. Thus, the biggest clique has to be greater or equal to the size of such clique. $\qed$
\end{proof}

\begin{lemma}\label{lemma:deg_fvector}
    Let $\mathcal{A}$ be a chromatic simplicial complex, the following equality holds:
    \[
    \max_{p \in \Pi}{\Delta(G_p(\mathcal{A}))} = \max_{p \in \Pi}{\max_{v \in V(\mathcal{A})}{\mathrm{f}_0(\mathrm{Lk}(\mathcal{A},\mathrm{St}(\mathcal{A},v))\restrict{p})}}
    \]
\end{lemma}
\begin{proof}
    Let $v_p\in V(\mathcal{A})$ be the vertex that gives the maximum argument on the right side of the equality, we have that $\mathrm{St}(\mathcal{A},v_p)$ gives the subcomplex of facets containing $v_p$. Observe that all vertices that are adjacent to $v_p$ are contained in the star. Then, we take the $\mathrm{f}_0$ of the link of the star, enumerating the vertices that are adjacent to the star (and then we restrict to a particular process label). Note that if we take a vertex $t \in \mathrm{St}(\mathcal{A},v_p)$, then if $w_p \in \mathrm{Lk}(\mathcal{A},\mathrm{St}(\mathcal{A},v_p)) \implies (v_p,w_p) \in G_p(\mathcal{A})$. But also if we have an edge $(v_p',w_p') \in G_p(\mathcal{A})$ it implies that it is in the link of the star. So finding the maximum link of the star of a vertex $v$ restricting to a process $p$ gives the maximum vertex degree of $G_p(\mathcal{A})$. $\qed$
\end{proof}

\subsection{$\mathrm{f}$-vector of the Standard Chromatic Subdivision}\label{sect:fvector_chroma}

The objective of this section is to derive equations enabling the computation of the $\mathrm{f}$-vector of the iterated chromatic subdivision using the $\mathrm{f}$-vector of the original simplicial complex. Initially, we present a general theorem that characterizes the $\mathrm{f}$-vector of any subdivision map. Subsequently, we demonstrate its instantiation for both the standard chromatic subdivision and open stars. The motivation behind these equations is to establish a method for computing the $\mathrm{f}$-vectors in Lemmas~\ref{lemma:clique_fvector} and~\ref{lemma:deg_fvector} based on the $\mathrm{f}$-vector of the initial simplicial complex.

\begin{theorem}[$\mathrm{f}$-vector of a subdivision] \label{th:fk_tauA}
Let $\tau$ be a subdivision operator, $\Delta^i$ the $i$-dimensional simplex, and $\mathcal{A}$ an $n$-dimensional simplicial complex. The following identity holds for the $\mathrm{f}$-vector of $\tau(\mathcal{A})$:
\[
\mathrm{f}_k(\tau(\mathcal{A})) = \sum_{i=k}^{n}{\mathrm{f}_i(\mathcal{A})\ \mathrm{f}_k(Int\ \tau(\Delta^i))}
\]
\end{theorem}
\begin{proof}
    Let $k$ be a proper coordinate of the $\mathrm{f}$-vector,$-1\leq k \leq n$,we will prove the identity by induction over the skeleton of $\mathcal{A}$. Let $S(i)=\mathrm{f}_k(\tau(Skel^i \mathcal{A}))$, we will prove that $S(i)=S(i-1)+\mathrm{f}_i(\mathcal{A})\cdot \mathrm{f}_k(Int\ \tau(\Delta^i))$. Note that $Skel^n \mathcal{A} = \mathcal{A}$. For the base case, observe that if $i<k$, then $\mathrm{f}_k(\tau(Skel^i \mathcal{A}))=0$. Because by definition of skeleton, $dim(Skel^i \mathcal{A}) \leq i$ and a subdivision operator cannot yield a simplicial complex of higher dimension than its input. For the inductive step, we define $K:= \tau(Skel^i \mathcal{A}) - \tau(Skel^{i-1} \mathcal{A})$, which is the set of $(i+1)$-tuples of $\tau(Skel^i \mathcal{A})$. Note that we use the term tuples because these are not proper simplices but subsets of $V(\mathcal{A})$ of size $i+1$.
    Then we have that $K \cup \tau(Skel^{i-1} \mathcal{A}) = \tau(Skel^i \mathcal{A})$ and $K \cap \tau(Skel^{i-1} \mathcal{A}) = \emptyset$. Thus, $\mathrm{f}_k(K) + \mathrm{f}_k(\tau(Skel^{i-1} \mathcal{A})) = \mathrm{f}_k(\tau(Skel^i \mathcal{A}))$.

    From the inductive hypothesis we have that $\mathrm{f}_k(\tau(Skel^{i-1} \mathcal{A})) = S(i-1)$. It remains only to enumerate $\mathrm{f}_k(K)$. First, note that all the $(i+1)$-tuples of $K$ come from the subdivision of an $i$-face of $\mathcal{A}$ and we have in total $\mathrm{f}_i(\mathcal{A})$ $i$-faces in $\mathcal{A}$. Thus, we have to enumerate how many $k$-faces are generated from subdividing the $i$-faces of $\mathcal{A}$. To know how many $i$-faces are yielded we apply $\tau$ to each $i$-face of $\mathcal{A}$. Let $\sigma \in Skel^i \mathcal{A}$ an $i$-face, note that $\sigma$ itself is a simplex of dimension $i$ - that is, $\Delta^i$.  It remains to observe that $\mathrm{f}_k(\tau(\sigma)) = \mathrm{f}_k(\tau ( \partial (\sigma))) + \mathrm{f}_k(\tau ( Int\ \sigma))$. Where $\mathrm{f}_k(\tau ( \partial (\sigma)))$ is already counted on $S(i-1)$, as $\tau ( \partial (\sigma)) \subseteq \tau(Skel^{i-1} \mathcal{A})$. This leaves that the new faces yielded are in $ Int\ \tau (\sigma^i)$. However, by definition of interior if we take $\sigma' \in Skel^i \mathcal{A}: \sigma' \neq \sigma$, we get $Int\ \tau(\sigma) \cap \tau(\sigma') = \emptyset$. Thus for the purposes of enumerating faces, we can replace $\sigma$ with $\Delta^i$.
    
    As we have $\mathrm{f}_i(\mathcal{A})$ $i$-faces, the new $i$-faces generated by the subdivision are $\mathrm{f}_i(\mathcal{A})\cdot \mathrm{f}_k(Int \ \tau (\Delta^i)) $. This concludes the inductive step. By solving the now proven recurrence relation, we get the final summation expression. $\qed$
\end{proof}

Theorem~\ref{th:fk_tauA} makes it possible to express the $\mathrm{f}$-vector of the subdivision of any simplicial complex in terms of its given $\mathrm{f}$-vector and the application of the subdivision operator on $\Delta^n$. Thus, if we can provide an expression for $\mathrm{f}_k(Int\ \tau(\Delta ^n))$, we can compute the subdivision for any simplicial complex. Note that we cannot use Theorem~\ref{th:fk_tauA} to compute the $\mathrm{f}$-vector of open stars since they do not form a simplicial complex. Thus, we present a reformulation for the specific case of open stars.

\begin{lemma} \label{lemma:f_k_tau_star}
Let $\tau$ be a subdivision operator, $\mathcal{A}$ a $n$-dimensional simplicial complex, $v \in V(\mathcal{A})$ and $r \in V(\Delta^i)$. The following identity holds for the $\mathrm{f}$-vector of $\mathrm{St}^\circ(\tau(\mathcal{A}),v)$:
\[
\mathrm{f}_k(\mathrm{St}^\circ(\tau(\mathcal{A}),v)) = \sum_{i=k}^{n}{\mathrm{f}_i(\mathrm{St}^\circ(\mathcal{A},v))\ \mathrm{f}_k(Int\ \mathrm{St}^\circ(\tau(\Delta^i),v))}
\]
\end{lemma}
\begin{proof}
    We apply can apply the same reasoning as in Theorem~\ref{th:fk_tauA}. However, as we are now enumerating the $k$-faces in an open star, some of the $i$-faces yielded in $Int \ \tau(\Delta^i)$ are not included in $\mathrm{St}^\circ(\tau(\Delta^i)$. Indeed, the ones that are in the open star are $Int \ \mathrm{St}^\circ(\tau(\Delta^i),r)$. Substituting  the new internal $k$-faces using this term in Theorem~\ref{th:fk_tauA} yields the desired expression. $\qed$
\end{proof}

\begin{lemma} \label{lemma:fIntStCh}
Let $v \in \Delta^n$ and $k>0$, the following identity holds:
\[
\mathrm{f}_k(Int\ \mathrm{St}^\circ(\mathrm{Ch}\ \Delta^n,v)) = \sum_{i=1}^{k}{\binom{n}{i} \mathrm{f}_{k-i}(\mathrm{St}^\circ(\mathrm{Ch}\ \Delta^{n-i},v))}
\]
\end{lemma}
\begin{proof}
We aim to enumerate the k-faces in $Int\ \mathrm{Ch}\ \Delta^n$ that include the vertex $v$. Moreover, as the faces are in the interior, they must necessarily include at least one vertex from the corresponding facet in the central simplex which has nodes with colors different from $v$. We note this facet as $\sigma$.

Hence, the number of faces will be all possible simplicial joins of $v$ with this facet and then with faces on the boundary of $\mathrm{St}^\circ(\mathrm{Ch}\ \Delta^n,v)$.

Let $k > 1$ (by definition of interior, if k = 0, we have $\mathrm{f}_0(Int\ \mathrm{St}^\circ(\mathrm{Ch}\ \Delta^n,v))=0$). Suppose we want to construct a $k$-face from an $(i-1)$-face of $\sigma$. To do this, we take a face $\delta \in \sigma : dim(\delta)=i-1$ and perform the simplicial join with $v$: $\delta * v$. Note that we can choose $\binom{n}{i}$ different $i-1$ faces. We perform the simplicial join $\delta * v$ and obtain an $i$-face in the interior. Now, if we want a $k$-face, we have no other option but to perform a second simplicial join with another face of dimension $(k-i)$ on $\partial\  \mathrm{St}^\circ(\mathrm{Ch}\ \Delta^n,v)$.

The faces we can join with $\delta$ are those that do not contain any of the colors of $\delta$; these are found in the subdivision of one of the $(k-i)$-faces of $\Delta^n$. Note that, as it is a subdivision, if we restrict ourselves to one of these faces, it is equal to the subdivision of a simplex of the dimension of such face. Therefore, the faces we can join are $\mathrm{f}_{k-i}(\Delta^{n-i})$. To obtain all faces, we have to enumerate from taking a single vertex from $\sigma$ ($dim(\delta * v)=1$) to taking a $k-1$ face ($dim(\delta * v)=k$), thus obtaining the formulated equation. $\qed$
\end{proof}

Note that Lemma~\ref{lemma:f_k_tau_star} applies to a generic simplicial complex. Consequently, we can use an already subdivided complex as input to obtain an expression for the iterative chromatic subdivision. Corollary~\ref{corollary:fk_chir_a} directly follows from Lemma~\ref{lemma:f_k_tau_star} and Lemma~\ref{lemma:fIntStCh}.

\begin{corollary} \label{corollary:fk_chir_a}
Let $r > 0$, $\mathcal{A}$ a simplicial complex, $v \in V(\mathcal{A})$, and $v' \in V(\Delta^i)$. The $\mathrm{f}$-vector of the open star in $v$ of the $r$-th chromatic subdivision of $\mathcal{A}$ is:
\[
\mathrm{f}_k(\mathrm{St}^\circ(\mathrm{Ch}^r\ \mathcal{A},v)) = \sum_{i=k}^{n}{\mathrm{f}_i(\mathrm{St}^\circ(\mathrm{Ch}^{r-1}\ \mathcal{A},v))\sum_{j=1}^{k}{\binom{i}{j} \mathrm{f}_{k-j}(\mathrm{St}^\circ(\mathrm{Ch}\ \Delta^{i-j},v'))}}
\]
\end{corollary}

Note that $\mathrm{f}_i(\mathrm{St}^\circ(\mathcal{A},v))$ represents an input value for the problem. Specifically, it denotes the $i$-faces containing the vertex $v$. A straightforward yet crucial observation is that we now have an expression to compute the $\mathrm{f}$-vector of the open star of $\mathrm{Ch}\ \Delta^n$:

\begin{corollary}\label{corollary:fk_st_chi_delta}
Let $v \in V(\Delta^n)$. The $\mathrm{f}$-vector of the open star in $v$ of $\mathrm{Ch}\ \Delta^n$ is:
\[
 \mathrm{f}_k(\mathrm{St}^\circ(\mathrm{Ch}\ \Delta^n,v)) = \sum_{i=k}^n{\binom{n}{i} \sum_{j=1}^k{\binom{i}{j} \mathrm{f}_{k-j}(\mathrm{St}^\circ(\mathrm{Ch}\ \Delta^{i-j},v))}}
\]
\end{corollary}
\begin{proof}
The result follows directly from Corollary~\ref{corollary:fk_chir_a}. Observe that to calculate $\mathrm{f}_k(Int\ \mathrm{St}^\circ(\mathrm{Ch}\ \Delta^n,v))$ we need to know $\mathrm{f}_{k'}(\mathrm{St}^\circ(\mathrm{Ch}\ \Delta^{n'},v))$ where $k' \in [0,k-1]$ and $n' \in [0,n-1]$, thus the recursion is well defined. $\qed$
\end{proof}

An interesting observation from Corollary~\ref{corollary:fk_st_chi_delta} is that $\{\mathrm{f}_n(\mathrm{St}^\circ (\mathrm{Ch}\  (\Delta^n,v))\}_{n\geq 0}$ is a well-known sequence known as ordered Bell numbers or Fubini numbers, which count the number of weak orderings on a set of $n$ elements~\cite{oeisFubini,fubiniNumbers}. 

We derive an equation to compute the $\mathrm{f}$-vector in Lemma~\ref{lemma:deg_fvector}, which is used to determine the vertex degree on the indistinguishability graphs.

\begin{lemma}\label{lemma:f_linkstar}
Let $p\in \Pi$, $\mathcal{A}$ a simplicial complex, $v_p \in V(\mathcal{A})\restrict{p}$. The following identity holds:
\[
\mathrm{f}_0(\mathrm{Ch}\ \mathcal{A}, \mathrm{Lk}(\mathrm{St}(\mathrm{Ch}\ \mathcal{A},v_p))\restrict{p}) = \sum_{i=1}^n{\mathrm{f}_i(\mathrm{St}^\circ(\mathcal{A},v_p)))}
\]
\end{lemma}
\begin{proof}
We want to count all vertices with label $p$ that are at distance 2 from $v_p$ in $\mathrm{Ch}\ \mathcal{A}$. $\mathrm{St}(\mathrm{Ch}\ \mathcal{A},v_p)$ gives the faces that contain $v_p$, and by counting the vertices in the link of the latter we are counting all vertices that are at distance two of $v_p$. We have that for each face of $\mathrm{St}^\circ(\mathcal{A},v_p)$, when subdivided by $\mathrm{Ch}$, will yield a single new vertex with label $p$. Moreover, such vertex will be at distance 2 from $v_p$, as $v_p$ will be adjacent to the interior vertices with different label than $v_p$ while these vertices will be adjacent to the interior vertex with label $p$. Moreover, all vertices of $\mathrm{Lk}(\mathcal{A},v_p)$ are at distance 2 from $v_p$ in $\mathrm{Ch}\ \mathcal{A}$, so the new vertices yielded are only the ones at distance 2 in $\mathrm{Ch}\ \mathcal{A}$. $\qed$
\end{proof}

Finally, we need to consider that, in Lemmas~\ref{lemma:deg_fvector} and~\ref{lemma:clique_fvector}, we search for the vertex with the biggest $\mathrm{f}$-vector at coordinates 0 and 1 in $\mathrm{Ch}^r\ \mathcal{A}$. Lemma~\ref{lemma:argmax_chst} shows that, for iterative subdivisions, the biggest $\mathrm{f}$-vector of an open star comes from the vertices that are at the first subdivision. This is a crucial result, as it allows us to fix the vertex with the biggest $\mathrm{f}$-vector and perform asymptotic analysis of the $\mathrm{f}$-vector on this vertex under iterative subdivisions.

\begin{lemma} \label{lemma:argmax_chst}
Let $\tau$ be a subdivision operator and $\mathcal{A}$ a n-dimensional simplicial complex.
\[
\argmax_{v \in V(\tau^r(\mathcal{A}))}{\mathrm{f}_k(\mathrm{St}^\circ(\tau^r(\mathcal{A}),v))} \in V(\tau(\mathcal{A})),\ \ \ \ \forall k\leq n, \forall r\geq 1
\]
\end{lemma}
\begin{proof}
We will proceed by induction on $r$ for an arbitrary $k$. The base step $r=1$ is trivial. For the inductive step, let $v=\argmax_{v \in V(\tau^r(\mathcal{A}))}{\mathrm{f}_k(\mathrm{St}^\circ(\tau^r(\mathcal{A}),v))}$. By inductive hypothesis, $v \in V(\tau(\mathcal{A}))$. We have that $\forall \Tilde{v} \in V(\tau^r(\mathcal{A})), \mathrm{f}_k(\mathrm{St}^\circ(\tau^r(\mathcal{A}),\Tilde{v}))\leq\mathrm{f}_k(\mathrm{St}^\circ(\tau^r(\mathcal{A}),v))$. From Lemma~\ref{lemma:f_k_tau_star}, we have that $\mathrm{f}_k(\mathrm{St}^\circ(\tau^{r+1}(\mathcal{A}),v))$ is computed as a sum of the faces in $\mathrm{St}^\circ(\tau^r(\mathcal{A}),v)$ multiplied by a constant that come from the simplex subdivision $\mathrm{f}_k(Int\ \mathrm{St}^\circ(\tau(\Delta^i),v))$. Thus, $\forall \Tilde{v} \in V(\tau^r(\mathcal{A})), \mathrm{f}_k(\mathrm{St}^\circ(\tau^{r+1}(\mathcal{A}),\Tilde{v}))$ $\leq \mathrm{f}_k(\mathrm{St}^\circ(\tau^{r+1}(\mathcal{A}),v))$. Next, the star of the new vertices $V(\tau^{r+1}(\mathcal{A}))\setminus V(\tau^r(\mathcal{A}))$ will have a constant $\mathrm{f}$-vector, irregardless of $r$, based on the dimension of the face they are a central simplex of. Hence, the vertex with the biggest $\mathrm{f}$-vector (in all its coordinates) of its open star will still be $v$. $\qed$
\end{proof}

\begin{remark}
It follows directly from Lemma~\ref{lemma:argmax_chst} that ${\max_{v \in V(\mathrm{Ch}^{r}\ \mathcal{A})}{\mathrm{f}_0(\mathrm{Lk}(\mathrm{Ch}^{r}\ \mathcal{A},\mathrm{St}(\mathrm{Ch}^{r}\ \mathcal{A},v)))\restrict{p}}} \in V(\mathrm{Ch}\ (\mathcal{A}))$ as Lemma~\ref{lemma:f_linkstar} shows that it is computed as the sum of open stars.
\end{remark}

Putting everything together, we can use the derived equations from Lemmas~\ref{lemma:deg_fvector} and~\ref{lemma:clique_fvector}, along with Lemma~\ref{lemma:argmax_chst}, to establish upper and lower bounds on the existence of an encoding set with size within these bounds:

\begin{theorem}\label{th:Imaencode_inequality}
    Let $\mathcal{I}$ a chromatic input complex and $r>0$. There exists a function $\textit{encode}: V(\mathcal{I})\rightarrow E$ such that when used as encoding function in the B-IIS protocol, $\Xi_b^r \cong \Xi^r$ and $|E|$ is bounded by:
    
    \[
    \max_{v \in V(\mathrm{Ch}\ \mathcal{I})}{\mathrm{f}_1(\mathrm{St}^\circ(\mathrm{Ch}^r\ \mathcal{I},v))} \leq |\Ima \textit{encode}| \leq \max_{v \in V(\mathrm{Ch}\ \mathcal{I})}{\mathrm{f}_0(\mathrm{Lk}(\mathrm{Ch}^r\ \mathcal{I},\mathrm{St}(\mathrm{Ch}^r\ \mathcal{I},v)))}+1
    \]
\end{theorem}

\section{Asymptotic Analysis on the $\mathrm{f}$-vector of the iterated chromatic subdivision} \label{sect:fvector_asympt} 

Recalling that our objective is to obtain asymptotic bounds on the size of the encoding set. Now, using the inequality of Theorem~\ref{th:Imaencode_inequality} and the $\mathrm{f}$-vector equations derived in Section~\ref{sect:fvector}, we can perform an asymptotic analysis w.r.t the number of processes in the system (directly corresponding to the dimension of the input complex) and the number of rounds $r$. Thus, our objective will be to find tight asymptotic bounds for $\mathrm{f}_1(\mathrm{St}^\circ(\mathrm{Ch}^r\ \mathcal{A},v))$ and $\mathrm{f}_0(\mathrm{Lk}(\mathrm{Ch}^r\ \mathcal{A},\mathrm{St}(\mathrm{Ch}^r\ \mathcal{A},v)))$.

We will use the conventional definitions of asymptotic analysis in computer science. If the reader is not familiar with this, we recommend checking~\cite{AlgDesign} for reference.

The following combinatorial identity will be used in the proof of Lemma~\ref{lemma:ThetaBoundFkdelta}:

\begin{lemma} \label{lemma:A1}
Let $n,k,r \in \mathbb{N}$ such that $r\leq r \leq n$, the following identity holds:
\[
\sum_{i=k}^{n}{\binom{n}{i}\binom{i}{r} b^{i-\alpha} (i-r)_{k-r} }
=
\frac{b^{k-\alpha}}{r!}(b+1)^{n-k}(n)_k
\]

Where $(n)_k$ is the falling factorial.
\end{lemma}
\begin{proof}
The result follows from the binomial theorem and the factorial formulation of the binomial coefficient:
  \begin{multline*}
    \sum_{i=k}^{n}{\binom{n}{i}\binom{i}{r} b^{i-\alpha} (i-r)_{k-r} } = \sum_{i=k}^{n}{\frac{n!}{(n-i)! r! (i-k)!} b^{i-\alpha} } 
    = \frac{n!}{r!} \frac{(n-k)!}{(n-k)!} \sum_{i=0}^{n-k}{\frac{b^{i+k-\alpha}}{(n-k-i)! i!} } \\
    = \frac{(n)_k}{r!} b^{k-\alpha} \sum_{i=0}^{n-k}{\binom{n-k}{i} b^i}
    = \frac{b^{k-\alpha}}{r!}(b+1)^{n-k}(n)_k \qed
  \end{multline*}
\end{proof}

\begin{lemma} \label{lemma:ThetaBoundFkdelta}
Let $\Delta^n$ be the $n$-dimensional simplex, $v \in V(\Delta^n)$, $(n)_k$ the falling factorial, and $k\leq n$, the following is a tight asymptotic bound on the $\mathrm{f}$-vector of the open star of $v$:
\[
\mathrm{f}_k(\mathrm{St}^{\circ}(\mathrm{Ch}\ \Delta^n,v) \in \Theta \bigg(\frac{(k+1)^{n-k} (n)_k }{\ln (2)^{k-1}}\bigg)
\]
\end{lemma}
\begin{proof}
First, to ease notation, let $T(k,n) := \mathrm{f}_k(\mathrm{St}^{\circ}(\mathrm{Ch}\ \Delta^n,v)) $, and define the bounding function $f(k,n) := \frac{(k+1)^{n-k} (n)_k }{\ln(2)^{k+1}}$. It is important to observe that in order to compute $T(k,n)$, we require the values ${T(k',n') : k'\in [0,k-1] \land n' \in [0,n-1]}$. Consequently, we can prove the result through a general induction over $n$ for all $k$.

The base case is trivial: $T(0,n) \overset{\mathrm{def}}{=} 1 \in \Theta(1) = \Theta(f(0,n)), \forall n \in \mathbb{N}$. For the induction step, we will prove it by computing the limit of $\frac{T(k,n)}{f(k,n)}$ as $n$ goes to infinity for all $k\leq n$ and showing this limit converges to a non-zero constant.  

\begin{multline*}
    \mathrm{f}_k(\mathrm{St}^{\circ}(\mathrm{Ch}\ \Delta^n,v))  \in \Theta(f(k,n)) \overset{\text{IH}}{\iff} \sum_{i=k}^{n}{\binom{n}{i} \sum_{j=1}^k {\binom{i}{j}\frac{ (k-j+1)^{i-k} (i-j)_{k-j}}{\ln (2)^{k-j-1}}}} \\
    \overset{\text{Lemma~\ref{lemma:A1}}}{=} \sum_{j=1}^{k}{\frac{\ln (2)^j}{\ln (2)^{k-1} j!} (k-j+2)^{n-k} (n)_k}
    =\frac{(n)_k}{\ln (2)^{k-1}}\sum_{j=1}^{k}{\frac{(k-j+2)^{n-k} \ln (2)^j}{j!}}  \in  \Theta(f(k,n)) \\
    \overset{\forall k\leq n}{\iff} \lim_{n \to \infty}{\frac{\frac{(n)_k}{\ln (2)^{k-1}}\sum_{j=1}^{k}{\frac{(k-j+2)^{n-k} \ln (2)^j}{j!}}}{(n)_k (k+1)^{n-k} \ln (2)^{-k+1}}} = C > 0
    \overset{\forall k\leq n}{\iff} \lim_{n \to \infty}{\sum_{j=1}^k{(\frac{k-j+2}{k+1})^{n-k} \frac{\ln (2)^j}{j!}}}=C>0 
\end{multline*}

We will now show that the limit of the last equation is in the interval $(\ln 2, 1)$ for all values of $k\leq n$. For the lower bound, noting that the term inside the sum is strictly positive, we take the first term of the sum: 
\[
\lim_{n \to \infty}{\sum_{j=1}^k{\big(\frac{k-j+2}{k+1}\big)^{n-k} \frac{\ln (2)^j}{j!}}} > \lim_{n \to \infty} (\frac{k+1}{k+1})^{n-k} \frac{\ln 2}{1!} = \ln 2 > 0 \ \ \forall{k\leq n}
\]
For the upper bound we use the fact that $\sum^{\infty}_{j=1}{\frac{\ln (2)^j}{j!}}=1$ and $\frac{k-j+2}{k+1} \leq 1 \ \forall k \leq n$. Thus the following series is a multiplication of a convergent sequence by another one which their values are in $(0,1]$:
\[
\lim_{n \to \infty}{\sum_{j=1}^k{\big(\frac{k-j+2}{k+1}\big)^{n-k} \frac{\ln (2)^j}{j!}}} < \lim_{n \to \infty}{\sum_{j=1}^{\infty}{\big(\frac{k-j+2}{k+1}\big)^{n-k} \frac{\ln (2)^j}{j!}}} < 1 \ \ \forall{k\leq n}
\]

Consequently, the asymptotic ratio between $\frac{T(k,n)}{f(k,n)}$ is in the interval $(\ln 2, 1)$ $\forall k\leq n$. It follows that $f(k,n)$ is a tight asymptotic bound of $T(k,n) \qed$

\end{proof}

Now that we have a tight asymptotic bound on $\mathrm{f}_k(\mathrm{St}^{\circ}(\mathrm{Ch}\  (\Delta^n),v))$, we can now give a bound for the equation in Lemma~\ref{lemma:fIntStCh}. Which will be useful for computing the final result regarding the $\mathrm{f}$-vector of the iterative chromatic subdivision:

\begin{lemma} \label{lemma:bound_intch}
Let $\Delta^n$ be the $n$-dimensional simplex, and $1 \leq k\leq n$,
\[
\sum_{j=1}^k{\binom{n}{j}\mathrm{f}_{k-j}(\mathrm{St}^{\circ}(\mathrm{Ch}\  \Delta^{n-j},v))} \in \Theta \bigg(\frac{(k+1)^{n-k} (n)_k }{\ln (2)^{k-1}}\bigg)
\]
\end{lemma}
\begin{proof}
Defining $T(k,n)$ as in Lemma~\ref{lemma:ThetaBoundFkdelta}, let $R(k,n):= \sum_{j=1}^k{\binom{n}{j}T(k-j,n-j)}$ and $f(k,n) := \frac{(k+1)^{n-k} (n)_k }{ln(2)^{k+1}}$.
\begin{multline*}
    R(k,n) \in \Theta(f(k,n)) \overset{\text{Lemma~\ref{lemma:ThetaBoundFkdelta}}}{\iff} \sum_{j=1}^{k}{\binom{n}{j} \frac{(k-j+1)^{i-k} (i-j)_{k-j}}{\ln (2)^{k-j-1}}} \\
    =\frac{(n)_k}{\ln (2)^{k-1}}\sum_{j=1}^{k}{\frac{(k-j+1)^{n-k} \ln (2)^j}{j!}}  \in  \Theta(f(k,n)
    \overset{\forall k\leq n}{\iff} \lim_{n \to \infty}{\sum_{j=1}^k{(\frac{k-j+1}{k+1})^{n-k} \frac{\ln (2)^j}{j!}}}=C>0 \ 
\end{multline*}
Where the previous limit is the same as in Lemma~\ref{lemma:ThetaBoundFkdelta} but with $k-j+1$ instead of $k-j+2$ on the numerator and $\frac{k-j+1}{k+1} < 1$. Thus the limit is non-zero convergent. $\qed$
\end{proof}

We are now ready to give a tight asymptotic bound on the iterative chromatic subdivision of a simplicial complex. 

\begin{theorem}[Asymptotic bound on the iterative chromatic subdivision] \label{theorem:finalBound}
Let $r>1$, $\mathcal{A}$ a $n$-dimensional simplicial complex and $v \in \mathcal{A}$. The following is a tight asymptotic bound on the $\mathrm{f}$-vector of the star of $v$ in $\mathrm{Ch}^r\ \mathcal{A}$:
\[
\mathrm{f}_k(\mathrm{St}^\circ(\mathrm{Ch}^r\ \mathcal{A},v)) \in \Theta \bigg(\bigg(\frac{n!}{\ln (2)^{n-1}}\bigg)^{r-1}\frac{(k+1)^{n-k}(n)_k}{\ln(2)^{k+1}} \bigg)
\] 
\end{theorem}
\begin{proof}
    Again, to relief notation we write $P(k,r) := \mathrm{f}_k(\mathrm{St}^\circ(\mathrm{Ch}^r\ \mathcal{A},v))$, $T$ and $R$ defined as in the proof of Lemma~\ref{lemma:bound_intch}. Then we denote the bounding function $f(k,r) := \bigg(\frac{n!}{\ln (2)^{n-1}}\bigg)^{r-1}\frac{(k+1)^{n-k}(n)_k}{\ln(2)^{k+1}}$. Note that $n$ is implicitly defined as the dimension of $\mathcal{A}$. Now instead of doing induction over $n$, we will prove it by induction over $r$ for all $k \leq n$.

    The base step: 
    \begin{align*}
        P(k,2) \in \Theta(f(k,2)) &\overset{\text{Lemma~\ref{lemma:bound_intch}}}{\iff} \sum_{i=k}^n{P(i,1) \frac{(k+1)^{i-k} (i)_k}{\ln (2)^{k-1}}} \in \Theta(f(k,2)) \\
        &\iff \sum_{i=k}^n{\frac{(i+1)^{n-i}(n)_i(k+1)^{i-k}(i)_k}{\ln (2)^{k-1} \ln(2)^{i-1}}} \in  \Theta(f(k,2))
    \end{align*}

    Note that for all $k$, the asymptotic dominating term of the sum w.r.t $n$ is the last term of the summation when $i=n$. Thus $P(k,2) \in \Theta(f(k,2))$. Now for the inductive step:

    \begin{multline*}
    P(k,r) \in \Theta(f(k,r)) \overset{\text{HI}}{\iff} \sum_{i=k}^n{f(k,r-1) \sum_{j=1}^k{\binom{i}{j} T(k-j,i-k) }} \in \Theta(f(k,r))  \\
    \overset{\text{Lemma~\ref{lemma:bound_intch}}}{=}\sum_{i=k}^n{f(i,r-1)\frac{(k+1)^{i-k} (i)_k}{\ln (2)^{k-1}}} 
    = \sum_{i=k}^n{\bigg(\frac{n!}{ln(2)^{i-1}}\bigg)^{r-2}\frac{(k+1)^{i-k} (i)_k}{\ln(2)^{k-1}}}  \in \Theta(f(k,r)) \\
    \end{multline*}

    Observe that the dominating term of the sum is again given by $i=n$ wich concludes the proof:

    \[
    \bigg(\frac{n!}{\ln(2)^{n-1}}\bigg)^{r-2}\frac{n!}{ln(2)^{n-1}}\frac{(k+1)^{n-k}(n)_k}{\ln(2)^{k+1}} = f(k,r)\ \ \qed
    \]

\end{proof}

An interesting fact is that for $\mathrm{f}_n(\mathrm{St}^{\circ}(\mathrm{Ch}\ \Delta^n,v))$, the following is an approximation of the Fubini numbers~\cite{fubiniNumbers}: $\mathrm{f}_n(\mathrm{St}^{\circ}(\mathrm{Ch}\  \Delta^n,v)) \sim \frac{n!}{2 \ln (2)^{n+1}}$, which corresponds with the tight asymptotic bound of Theorem~\ref{theorem:finalBound}.

The final bound needed to compute the upper limit of the size of the encoding set follows directly from Theorem~\ref{theorem:finalBound}.

\begin{corollary} \label{coro:bound_linkstar}
Let $\mathcal{A}$ be a $n$-dimensional simplicial complex and $v \in V(\mathcal{A})$:
\[
\mathrm{f}_0(\mathrm{Lk}(\mathrm{Ch}^r\ \mathcal{A},\mathrm{St}(\mathrm{Ch}^r\ \mathcal{A},v))) \in \Theta \bigg(\bigg(\frac{n!  }{\ln (2)^{n-1}}\bigg)^r n^n \bigg)
\]
\end{corollary}
\begin{proof}
    As before, to ease notation let $f_n(k,r):=\bigg(\bigg(\frac{n! n^n }{\ln (2)^{n-1}}\bigg)^r \bigg)$ be the bounding function.
    \begin{multline*}
        \mathrm{f}_0(\mathrm{Lk}(\mathrm{Ch}^r\ \mathcal{A},\mathrm{St}(\mathrm{Ch}^r\ \mathcal{A},v))) \in \Theta(f_n(k,r)) \overset{\text{Lemma~\ref{lemma:f_linkstar}}}{\iff} \sum_{i=1}^n{\mathrm{f}_i(\mathrm{St}^\circ(\mathrm{Ch}^r \mathcal{A},v))} \in \Theta(f_n(k,r)) \\
        \overset{\text{Theorem~\ref{theorem:finalBound}}}{\iff} \bigg(\frac{n!}{\ln(2)^{n-1}}\bigg)^{r-1}\sum_{i=1}^n{\frac{(i+1)^{n-i} (n)_i}{\ln(2)^{i-1}}} \in \Theta(f(k,r)) 
        \iff \lim_{n \to \infty}{\frac{\big(\frac{n!}{\ln(2)^{n-1}}\big)^{r-1}\sum_{i=1}^n{\frac{(i+1)^{n-i} (n)_i}{\ln(2)^{i-1}}}}{\big(\frac{n!  }{\ln (2)^{n-1}}\big)^r n^n }} \\
        =  \lim_{n \to \infty}{\sum_{i=1}^n{ \frac{\ln(2)^{n-i}}{(n-i)!} \cdot \bigg(\frac{i+1}{n}\bigg)^n \cdot \frac{1}{(i+1)^i}}}=C>0
    \end{multline*}

Observe that in the final equation, the sum can be separated by 3 independent sequences, where each one of them is positive and convergent. Moreover, these sequences are multiplied to the power of $r$. As all terms are strictly positive, the sum will be grater than 0. And it will converge as the series of multiplied positive sequences are convergent. $\qed$.
\end{proof}

\section{Final Results} \label{sect:final results}

Using the asymptotic bounds calculated in Section~\ref{sect:fvector_asympt}, we can establish both upper and lower bit complexity asymptotic bounds for the B-IIS protocol.

\begin{corollary}\label{corollary:final_lowerbound}
Let $\Xi^r$ and $\Xi^r_b$ be the protocol maps of the full information and bounded IIS protocols for $r$ rounds, respectively. Let $n>2$ be the number of processes in the system, and $\mathcal{I}$ be a chromatic input complex. If $\Xi^r(\mathcal{I}) \cong \Xi_b^r(\mathcal{I})$, then $\Omega(rn\log n)$ is a lower bound on the bit complexity of the $r$-th round of the B-IIS protocol.
\end{corollary}
\begin{proof}
    If $\Xi^r(\mathcal{I}) \cong \Xi_b^r(\mathcal{I})$, by Corollary~\ref{coro:equiv_B-IIS_iterated}, we need $\mathrm{Ch}\ \mathcal{I}$ to be distinguishable under the encode function associated to $\Xi_b$. Then by Theorem~\ref{th:Imaencode_inequality} we need $|\Ima \textit{encode}| \geq \max_{v \in V(\mathrm{Ch}\ \mathcal{A})}{\mathrm{f}_1(\mathrm{St}^\circ(\mathrm{Ch}^r\ \mathcal{A},v))}$. From Theorem~\ref{theorem:finalBound}, we have that $\mathrm{f}_1(\mathrm{St}^\circ(\mathrm{Ch}^r\ \mathcal{A},v)) \in \Theta \big( n!^r 2^{n-1}n \big)$. We require that the function $\textit{encode}$ be capable of encoding each of these values in an entry of $M[k]$. Therefore, we need at least $\Theta(\log_2 ((n!^r 2^{n-1}n ))$ bits to encode all those values, hence the bit complexity is at least $\Omega(rn \log n)$. $\qed$
\end{proof}

\begin{corollary}\label{coro:final_upperbound}
Let $\Xi^r$ and $\Xi^r_b$ be the protocol maps of the full information and bounded IIS protocols for $r$ rounds, respectively. Let $n>2$ be the number of processes in the system, and $\mathcal{I}$ be a chromatic input complex. There exists an encoding function for $\Xi^r_b$ with bit complexity $O(r n \log n)$ such that $\Xi^r(\mathcal{I}) \cong \Xi_b^r(\mathcal{I})$. 
\end{corollary}
\begin{proof}
 By Theorem~\ref{th:equiv_xib}, we know that there exists a function $\textit{encode}$ such that $\Xi_b^r \cong \Xi^r$ and $|\Ima \textit{encode}| \leq \max_{v \in V(\mathrm{Ch}\ \mathcal{A})}{\mathrm{f}_0(\mathrm{Lk}(\mathrm{Ch}^r\ \mathcal{A},\mathrm{St}(\mathrm{Ch}^r\ \mathcal{A},v)))}+1$. From Lemma~\ref{coro:bound_linkstar}, we have that $\textit{encode}$ will require to write at most $\Theta (\log_2(n!^r n^n)$ bits on an entry of $M[k]$. Which gives $O(r n \log n)$ as an upper bound on the bit complexity. $\qed$
\end{proof}

Observe that for the case of a two-process system, Theorem~\ref{theorem:finalBound} and Lemma~\ref{coro:bound_linkstar} give a tight bound $\Theta(1)$ on the bit complexity. In Appendix~\ref{sect:approx_ag} this is further studied by giving an algorithm that solves $\epsilon$-agreement for arbitrarily small $\epsilon$ with $2$ bit complexity.

It is noteworthy that the asymptotic complexity does only depend on the number of processes and iterations of immediate snapshot. The input size (the number of simplices in the input complex) plays only a constant factor when calculating the $\mathrm{f}$-vector over iterative subdivisions. That is because with each subdivision, the number of indistinguishable states only increases \emph{locally} within each simplex.

\subsubsection{Conclusions and open questions.}
This paper appears to be the first to address simulations of \textit{full-information} protocols using bounded iterated memory. 
%
By using tools of combinatorial topology, we derive a tight bound on the bit complexity required to simulate IIS. 
Our results extend and complement recent findings on two-process systems~\cite{2process_complexity,delporte2023computational}, shedding light on the feasibility of simulating IIS with a constant bit complexity. 
Our results underscore the practical application of combinatorial topology for distributed systems by identifying necessary and sufficient properties that the protocol complex must possess in order to yield a desired result. 

Future lines of work involve using the same algebraic tools to characterize different protocol maps and complexes. 
This could lead to optimal bit complexity protocols to solve different kind of tasks. 
For instance, we have that for colorless tasks, the chromatic subdivision can be replaced by the simpler and more studied barycentric subdivision, which has a different $\mathrm{f}$-vector. 

Furthermore, this work focused on simulating iterated protocols while preserving the given number of iterations. It is an open problem studying the communication complexity when relaxing this condition, permitting the use of iterated protocols to solve a single round of IIS. Recent research has demonstrated the feasibility of simulating a round of IIS with 2-bit entries, albeit at the cost of an exponential increase in the number of rounds \cite{delporte2023computational}. However, the correlation between bit complexity and the number of iterations required to simulate an IIS round remains an unanswered question.

\subsubsection{Acknowledgments.} Guillermo Toyos-Marfurt was supported by Uruguay's National Agency for Research and Innovation (ANII) under the code “POS\_EXT\_2021\_1\_171849”.
Petr Kuznetsov was supported by TrustShare Innovation Chair (financed by Mazars Group). 


\bibliographystyle{splncs04}
\bibliography{biblio,references}{}

\appendix 

\section{Bit complexity of Approximate Agreement} \label{sect:approx_ag}

In the 2 process approximate agreement task, both processes have an input value and have to output values that are at distance $\frac{1}{\epsilon}$ from each other. If a process doesn't hear from the other, it has to output its input. The problem was first presented in~\cite{aproxag} and it has been shown that the task is solvable on the IIS model~\cite{distCompTopo}.

For simplicity, we address an inputless version of approximate agreement where one process has input $1$ and the other $0$. In terms of $(\I,\O, \Delta)$, we have that $\mathcal{I}=\{0,1,\{0,1\}\}$. W.l.o.g, we assume process $p_0$ has input $0$ and $p_1$ input $1$. Then the output graph consists of a path of $\epsilon$ edges. Whose vertices are $V(\mathcal{O})=\{(p_0,0),(p_1,\frac{1}{\epsilon}),(p_0,\frac{2}{\epsilon}),\dots,(p_0,\frac{\epsilon-1}{\epsilon}),(p_1,1)\}$. And the task specification is defined by the carrier map $\Delta((p_0,0))={(p_0,0)}$, $\Delta((p_1,1))={(p_1,1)}$ and $\Delta(\{(p_0,0),(p_1,1)\})=\mathcal{O}$.

Notice that if $\epsilon=0$, then the task is equivalent to the consensus problem, which has been shown to be impossible to solve in a Wait-Free asynchronous system~\cite{consensus}. The $\epsilon$-agreement problem was first presented, together with a solution, in~\cite{aproxag}.

Using the IIS model, we can achieve approximate agreement with precision $\frac{1}{\epsilon}$ using $\log_3(\ceil{\frac{1}{\epsilon}})$ layers~\cite{distCompTopo}. Interestingly, it is sufficient to have two bits per entry on the iterated memory to solve approximate agreement regardless of $\epsilon$. The key intuition is illustrated in figure~\ref{fig:b2l3}. On each round, a process needs to determine to which state to move according to its view. Thus, except for the rightmost and leftmost vertices, which can receive one message or none, a process can receive two different messages or no message at all. 

\begin{figure}
    \centering
    \includegraphics[width=0.8\columnwidth]{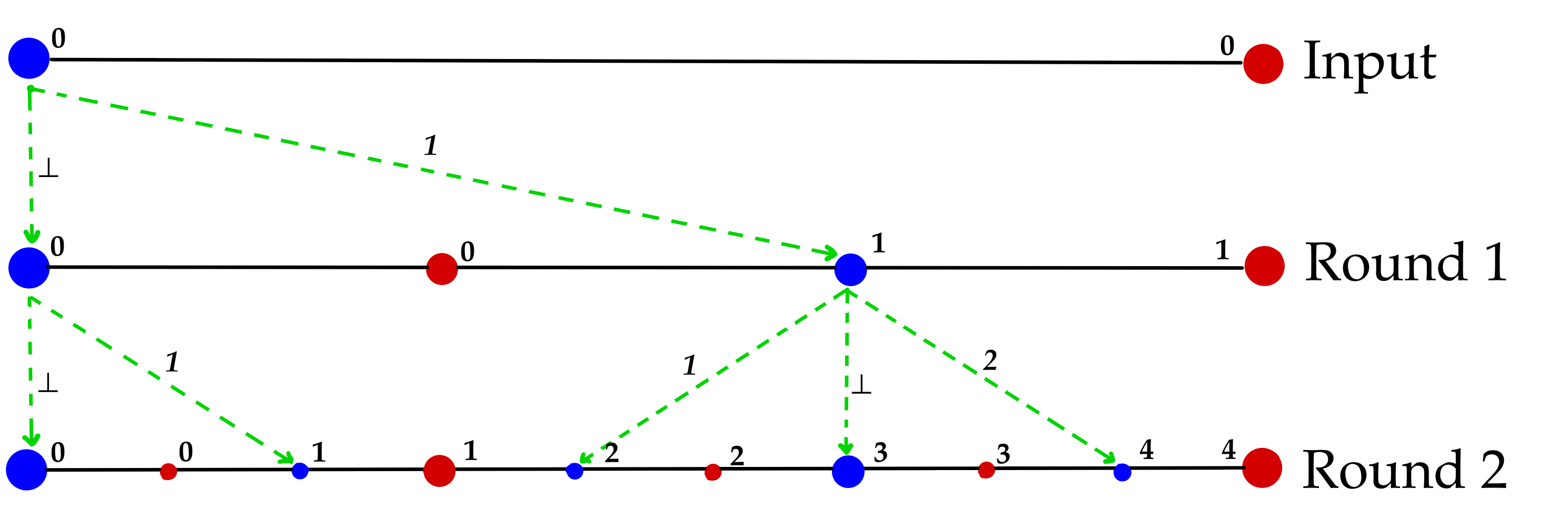}
    \caption{From top to bottom, the input complex, protocol complex for one round, and the protocol complex for two rounds of the bounded approximate agreement algorithm. Numbers next to the vertices indicate the internal state of the process. Green arrows illustrate the next_state function executed in each round for the blue process.} 
    \label{fig:b2l3}
\end{figure}

Next, we present an approximate agreement protocol with 2 bit complexity. Subsequently, we define the decision function $\delta$, the encoding function $\textit{encode}$ and the $\textit{next\_state}$ function that encapsulates the internal state of the process at each layer. Recalling the structure of the B-IIS algorithm in Algorithm~\ref{alg:B-IIS}, let the state space $\mathcal{S} \subset \mathbb{N}$ and the next-state function $next\_state : \mathcal{S} \rightarrow \mathcal{S}$ as defined in algorithm~\ref{alg:nextState}. Then, the encoding function alternates which message to write: $\textit{encode}: \mathcal{S} \rightarrow \{1,2\},\  \textit{encode}(s)=1+(s \mod 2)$. Finally, the decision map $\delta_i : \mathcal{S} \rightarrow \{\frac{k}{\epsilon} : k \in [0,\epsilon]\},\ \delta_i(s)=\frac{2s + i}{\epsilon}$. We call the protocol map of this algorithm $\Xi_\epsilon$.

\begin{algorithm}
\SetAlgoLined
\If{$s=\bot$}{
    \Return $0$;
}
\If{$v[(1-i)]=2$}{
    \Return $3s+i+(-1)^s(2i-1)$
}\If{$v[(1-i)]=\bot$}{
    \Return $3s + i$
}\If{$v[(1-i)]=1$}{
    \Return $3s + i +(-1)^s(1-2i)$
}
\caption{$next\_state_i(v,s)$ function}\label{alg:nextState}
\end{algorithm}

\begin{theorem}
    The approximate agreement task with precision $\frac{1}{\epsilon}$ is solvable $\forall \epsilon > 0$ in the B-IIS model. Moreover, the required number of iterations is log-proportional to the precision: $r \geq log_3(\epsilon)$.
\end{theorem}
\begin{proof}
    We need to prove that $\Xi_\epsilon(\mathcal{I})$ yields a protocol complex consisting of a path of $\epsilon$ edges, as in figure~\ref{fig:b2l3}, made out of vertices $(s,p_{\{0,1\}})$ where $s$ is the internal process state. In $\Xi_\epsilon(\mathcal{I})$, any two adjacent vertices $(s,p_0)$ and $(s',p_1)$ we have that $s-s'\leq 1$. W.l.o.g we orient the protocol complex such that the leftmost vertex is $(0,p_0)$ and the rightmost $(\frac{3^r-1}{2},p_1)$. 

    First, observe that if the protocol map $\Xi_\epsilon(\mathcal{I})$ behaves as described, then the $\textit{encode}$ function is correct: Take any pair of adjacent vertices $(s,p_0)$ and $(s',p_1)$, we have that $|\delta_0(s)-\delta_1(s')|=|\frac{2(s-s')-1}{\epsilon}|\leq \frac{1}{\epsilon}$. Moreover, $\delta_0(0)=0$ and $\delta_1(\frac{3^r-1}{2})=\frac{3^r}{\epsilon}=1$. From here it follows that to have $\frac{1}{\epsilon}$ precision, we need at least that $r\geq log_3(\epsilon)$.
    
    Now we will prove by induction over the number of iterations $r$ that $\textit{encode}$ and $next\_state_i$ yield the desired protocol complex. The base case, $r=0$ is trivial: $\Xi_\epsilon^0(\mathcal{I})=\mathcal{I}$ where $\mathcal{I}$ is a path of length 1, and $\delta_0(0)=0$, $\delta_1(0)=1$. For the inductive step, because $s-s'\leq 1$, we distinguish two possible process configurations: $\{(s,p_0),(s,p_1)\}$ and $\{(s+1,p_0),(s,p_1)\}$. Assume $s\geq0$ even, thus $\textit{encode}(s)=1$.

    Consider the edge $\{(s,p_0),(s,p_1)\}$. In the next iteration, there are 3 possibilities: Both processes read $2$, $p_0$ reads $\bot$ but $p_1$ reads $2$, and $p_1$ reads $\bot$ but $p_0$ reads $2$. The states in the next iteration will be $next\_state_0(1,s)=3s+1$, $next\_state_1(1,s)=3s$, $next\_state_0(\bot,s)=3s$ and $next\_state_1(\bot,s)=3s+1$. Thus, from $\{(s,p_0),(s,p_1)\}$ the next iteration will yield a path of 3 edges: $(3s,p_0)-(3s,p_1)-(3s+1,p_0)-(3s+1,p_1)$. We have that the states of $p_0$ are equal or at most $1$ value bigger than the states of $p_1$. The reasoning is analog if $s$ is odd.

    Thus, each edge of the $r-1$ iteration protocol complex yields 3 new edges in the next iteration, where the initial vertices are relabeled by the function $next\_state_i(s,\bot)=3s+i$. As a result, $\Xi_\epsilon^r$ will be a path of length $3^r$ such that leftmost vertex is $(0,p_0)$ and the rightmost $(\frac{3^r-1}{2},p_1)$. $\qed$
\end{proof}

It is remarkable that we can have unbounded $\epsilon$ precision using constant bit complexity. Furthermore we need $\log_3(\epsilon)$ layers to reach the desired precision. Therefore, we can conclude that we can reach $\epsilon$ approximate agreement in $O(log(\epsilon))$ layers where each shared-memory entry has 2 bits.

The B-IIS instance introduced here resembles the approximate agreement algorithm proposed by Delporte-Gallet et al.~\cite{2process_complexity} for wait-free dynamic networks, where primarily 1-bit messages are utilized. However, in B-IIS, 2 bits are required due to the context of shared memory systems, requiring the encoding of three different values: those written by the processes ($1$ and $2$), along with an initial value, denoted here as $\bot$, representing a clean state of a shared memory entry.

Furthermore, the colored approximate agreement task's output complex is isomorphic to $\lceil \log_3 (\epsilon) \rceil$ applications of the standard chromatic subdivision~\cite{distCompTopo}. Therefore, the protocol map $\Xi_\epsilon$ is an instance of B-IIS equivalent, up to isomorphism, to the full-information IIS protocol but with a constant 2-bit complexity.

In conclusion, for 2 processes to replicate the \textit{full-information} IIS protocol the bit complexity is 2 ($\Theta(1)$). Observe that this optimal asymptotic bound can be derived from the asymptotic calculated in section~\ref{sect:fvector_asympt}. As for a 1-dimensional simplex $\mathrm{f}_1(St^\circ(\mathrm{Ch}^r\ \Delta^1,v)) \in \Theta(1)$. 

\end{document}